\documentclass{toc}

\tocdetails{%
  title = {Even Quantum Advice is Unlikely to Solve $\PP$},
plaintexttitle = {Even Quantum Advice is Unlikely to Solve PP}, %
  number_of_pages = {NN},
  number_of_bibitems = {NN},
  number_of_figures = {NN},
  conference_version = {XXXXNN}, %
  author = {Justin Yirka},
  authorlist = {Justin Yirka},
  acmclassification = {F.1.3},
  amsclassification = {68Q15, 68Q12, 68Q06},
  keywords = {circuit complexity, quantum computing, counting hierarchy, quantum advice, oblivious proofs},
}   %

\tocdetails{%
volume = 21,      %
number = 7,       %
year = 2025,
received = {May 29, 2024}, %
revised = {September 4, 2025},
published = {October 2, 2025},
note,
doi = {10.4086/toc.2025.v021a007},
}   %

\usepackage{packages}

\usepackage{aumacros}

\begin{document}

\begin{frontmatter}%

\title{Even Quantum Advice is Unlikely to Solve \texorpdfstring{$\PP$}{PP}}  %

\author[yirka]{Justin Yirka\thanks{Supported via Scott Aaronson by a Simons Investigator Award.}}

\begin{abstract}
	We give a corrected proof that if
    $\mathsf{PP} \subseteq \mathsf{BQP\!/\hspace{-0.11em}qpoly}$ %
		(probabilistic polynomial time can be efficiently simulated by
		quantum circuits with quantum advice), %
        then the Counting Hierarchy collapses,
	as originally claimed %
        by Aaronson (CCC'06). %
	This recovers the related unconditional claim that %
	$\mathsf{PP}$ does not have circuits %
	of any fixed-polynomial %
        size $n^k$ even with quantum advice.
Our result is based on  %
proving that %
$\mathsf{YQP^*}$,
an oblivious version of %
$\mathsf{QMA}\cap\mathsf{coQMA}$, %
is contained in %
$\mathsf{APP}$,
a $\mathsf{PP}$-low subclass of $\mathsf{PP}$ with an arbitrarily small but nonzero promise gap.
\end{abstract}

%

\end{frontmatter}

\section{Introduction}

Do reasonably-sized circuits solve hard problems, 
given they may be chosen non-uniformly for each input size?
While directly answering this question for classes such as $\NP$ has proven difficult,
progress has been made showing conditional results,
such as the Karp--Lipton theorem that if $\NP\subseteq \Ppoly$,
then the Polynomial Hierarchy collapses \cite{karp1980some},
or showing upper bounds against larger classes,
such as $\PP$ or $\NEXP$ \cite{vinodchandran2005note,williams2014nonuniform}.
Exploring further, we can consider quantum computation.
In this model, circuits are typically uniformly generated
but might accept non-uniform \emph{advice} strings as part of their input.
Moreover, quantum circuits can receive not only classical
advice strings ($\BQPpoly$), but quantum \emph{advice states} ($\BQPqpoly$).

In \cite{aaronson_subtle}, Aaronson proved new quantum circuit lower bounds,
among other results.
In particular, he gave several results characterizing quantum circuits' ability
to solve problems in the class $\PP$,
the class of problems decidable by probabilistic 
polynomial-time %
algorithms with no promise gap,
\ie{}, which accept with probability at least $1/2$ or strictly less than $1/2$.
$\PP$ contains both $\BPP$ and $\NP$ and is contained in $\PSPACE$.
Aaronson proved that $\PTIME^\PP$ does not have circuits of size $n^k$ for any
fixed constant $k$ even if the circuits use quantum advice states.
Second, he claimed a quantum analogue of the Karp--Lipton theorem, showing that if
$\PP\subseteq \BQPqpoly$, then the Counting Hierarchy ($\CH$) collapses to $\QMA$,
where the Counting Hierarchy is the infinite sequence of classes $\CHi[1]=\PP$ and
$\CHi[i]=\left(\CHi[i-1]\right)^\PP$,
and where $\QMA$ is a quantum analogue of $\NP$.
Similarly, he showed that under the stronger assumption $\PP \subseteq \BQPpoly$,
using classical advice instead of quantum advice,
then $\CH = \QCMA$.
Third, Aaronson combined these results to give the unconditional bound that $\PP$
does not have classical or quantum circuits of size $n^k$ for
any fixed constant $k$ even with quantum advice.\footnote{Slightly earlier,
	Vinochandran \cite{vinodchandran2005note} gave a proof that $\PP$ does not
	have \emph{classical} circuits of %
	fixed-polynomial size.} %

However, Aaronson later noted there was an error in one of the proofs \cite{aaronsonBlogErrors2017}.
The first of the above results was unaffected,
but the proof of the second result only held under the stronger assumption that
$\PP\subseteq \BQPpoly$.
This also meant the third result only held for quantum circuits with classical,
not quantum, advice.
Fortunately, no other results in \cite{aaronson_subtle} were affected,
but no fix for this bug
was forthcoming.

Very briefly, the error was a claim that for oracle classes of the form
$\class{C}^{\BQPqpoly}$, if a machine for the base class $\class{C}$
is able to find the quantum advice state that will be used by the oracle machine,
then the base machine can ``hard-code'' the advice state into its oracle queries so that
the oracle no longer needs the power to find its own advice,
thus reducing $\class{C}^{\BQPqpoly}$ to $\class{C}^{\BQP}$.
This approach works for classes with classical advice, like $\class{C}^{\BQPpoly}$.
But, because complexity classes and their associated oracles
are defined in terms of (classical) strings as input,
there is no way to hard-code a general quantum advice state into a query.

In this note, we give a corrected proof of Aaronson's full claims.
We show that if $\PP\subseteq \BQPqpoly$,
then the Counting Hierarchy collapses to $\QMA$ and in fact to $\YQPstar$, defined below.
Given this correction, Aaronson's proof for the third claim,
that $\PP$ does not have circuits of size $n^k$
for any fixed constant $k$ even with quantum advice, now goes through.

Our primary technical contribution is to show $\YQPstar \subseteq \APP$.
Here, $\YQPstar$
(\expref{Definition}{def:YQP}) %
is an oblivious version of $\QMA\cap\class{coQMA}$, 
where ``oblivious'' means
there exists a useful proof state which depends only on the size of the input.
Crucially, a $\YQPstar$ protocol includes a proof-verification circuit that tests
if the given quantum state is a ``good'' proof, independent of whether
a particular input is a YES- or NO-instance.
The class $\APP$  %
(\expref{Definition}{def:APP}) %
is a subclass of $\PP$ with an arbitrarily small but nonzero promise gap.
It is known to have the nice property $\PP^\APP = \PP$ \cite{li1993counting}, \ie{}, it is $\PP$-low.
Thus, we find that $\YQPstar$ is also $\PP$-low.
Our corrected proof combines this result with
the equality $\BQPqpoly = \YQPstarpoly$,
serendipitously proven by Aaronson with Drucker \cite{aaronson2014full}.
Now, instead of following Aaronson's original attempt to collapse $\PP^\PP$ to $\PP^\BQPqpoly$ to $\PP^\BQP$ to $\PP$,
we can collapse $\PP^\PP$ to $\PP^\YQPstarpoly$ to $\PP^\YQPstar$ to $\PP$.

Our results provide stronger implications and improved bounds for quantum circuits
with quantum advice
and establish new insights into
$\PP$-lowness and classes within $\APP$.
Known quantum Karp--Lipton style bounds include that if
$\NP\subseteq \BQPqpoly$, then
$\Pii[2] \subseteq \QMA^{\class{PromiseQMA}}$ \cite{aaronson2014full},
and that if
$\QCMA\subseteq \BQPpoly$, then $\class{QCPH}$ collapses
to its second level \cite{agarwal2024quantum},
where $\class{QCPH}$ is defined like $\PH$ but with quantum verifiers and classical proofs.
Our result that if $\PP\subseteq \BQPqpoly$, then $\CH = \YQPstar$ adds to this list.
As for unconditional bounds, following Aaronson's unaffected result that
$\PTIME^\PP$ does not have quantum circuits with quantum advice of any %
fixed-polynomial %
size, our corrected result that $\PP$ also does not have such circuits is the first improved bound
on %
fixed-polynomial-size  %
circuits with quantum advice.
Regarding $\PP$-lowness, our primary lemma establishes $\YQPstar$ as the largest
natural quantum complexity class known to be $\PP$-low,
improving on the fact that $\BQP$ is $\PP$-low \cite{fortnow1999complexity}.\footnote{Morimae and Nishimura \cite{morimae16}
	gave definitions involving quantum postselection constructed to equal $\AWPP$
	and $\APP$, which are $\PP$-low.}
Finally, our result adds to the verification-related classes known to be contained in $\APP$,
including $\class{FewP}$ \cite{li1993counting} (but not $\NP$),
showing that $\APP$ contains the oblivious-witness classes
$\YQPstar\supseteq\class{YMA^*}\supseteq\class{YP^*}$.

\section{Preliminaries}

In this section, we give definitions, state a fact relating quantum circuits to $\GapP$,
and recall the technique of in-place error reduction
to make a few technical observations that will be useful for our main result.
For a deeper introduction to this area,
see the textbook by Arora and Barak \cite{arora2009computational} or the survey by Watrous \cite{watrous2008quantum}.
For more motivation, see
\cite{aaronson_subtle,aaronson2014full}.

\paragraph{Non-uniform circuits}
The classes of non-uniform circuits we discuss, such as $\Ppoly$, $\BQPpoly$, and $\BQPqpoly$,
share the following key characteristics.
First, the circuits are defined with bounded fan-in and fan-out, in contrast to
classes such as $\class{AC}_0$ or $\class{QAC}_0$.
Second, the classes consider circuits of polynomial-size,
where the size is the number of gates in a circuit.
Third, they are defined in terms of circuits or advice that may depend
on the size of the problem input (but not on the input itself),
with no requirement that the circuit or advice is generated by a uniform algorithm.

Advice is generally considered ``trusted'', in contrast to the
untrusted proofs or witnesses received in classes such as $\NP$.
In $\NP$, in a NO-instance, a verifier should not accept given any proof.
But in $\Ppoly$, a circuit is only guaranteed to be correct when given the correct advice.

The names $\BQPpoly$ and $\BQPqpoly$ have been used to refer to similar but distinct classes
(cf., $\BQPmpoly, \BQPQpoly,$ and the Complexity Zoo \cite{ZooBQPmpoly}).
The distinction is in the behavior of the circuit when the ``wrong'' advice is provided:
is a probabilistic circuit required to accept with high or low probability
(outside of the promise gap) only when the correct advice is provided, or for all advice? 
We follow the convention that because
advice is considered ``trusted'', there is no need for a promised behavior
when given the wrong advice.
An explicit definition of $\BQPqpoly$, following this convention, is given below.
We follow this same convention for $\BQPpoly$ and, later, $\YQPstarpoly$.

\begin{definition}[\cite{aaronson2014full}]
	A language $L$ is in $\BQPqpoly$ if there exists a polynomial-time quantum algorithm $A$
	and polynomial-time computable function $m\leq \poly(n)$ such that for all $n$,
	there exists an $m$-qubit advice state $\rho_n$ such that $A(x,\rho_n)$ outputs $L(x)$
	with probability at least $2/3$ for all $x\in\{0,1\}^n$.
\end{definition}

\paragraph{\texorpdfstring{$\mathsf{\mathbf{YQP}}$}{YQP}}
The class $\YQP$ was first described in \cite{aaronson2007learnability},
but the definition was later corrected by Aaronson and Drucker \cite{aaronson2014full}.
Informally, it is the oblivious version of $\QMA{}\cap\class{coQMA}$,
where ``oblivious''
means that the witness sent by the prover, Merlin, depends only on the length of the input.
Oblivious proofs can also be thought of as a restriction of non-uniform classes, like $\BQPqpoly$,
to advice that is verifiable, as in $\NP$ and $\QMA$ \cite{goldreich2015input}.
This combination in $\YQP$ been described as
``untrusted advice'' \cite{aaronson2007learnability}.

\begin{definition}\label{def:YQP}
	A language $L$ is in $\YQP$ if there exists a polynomial-time uniform family of
	quantum circuits $\{Y_n\}_{n\in \N}$ that satisfy the following.
	Circuit $Y_n$ is of size $\poly(n)$ and takes as input
	$x\in \{0,1\}^n$, an $m$-qubit state $\rho$ for some $m\in \poly(n)$,
	and an ancilla register initialized to the all-zero state,
	and has two designated ``advice-testing'' and ``output'' qubits.
	$Y_n(x,\rho)$ acts as follows:
	\begin{enumerate}
		\item First, $Y_n$ applies a subcircuit $A_n$ to all registers, after which
		the advice-testing qubit is measured, producing a value $\badv\in\{0,1\}$.
		\item Next, $Y_n$ applies a second subcircuit $B_n$ to all registers,
		then measures the output qubit, producing a value $\bout\in\{0,1\}$.
	\end{enumerate}
	These output bits satisfy the following:
	\begin{itemize}
		\item For all $n$, there exists a $\rho_n$ such that for all $x \in \{0,1\}^n$,
		the advice bit satisfies $\expec[\badv] \geq 9/10$.
		\item For any $x,\rho$ such that $\expec[\badv]\geq 1/10$,
		on input $x,\rho$ we have
		\[
		\Pr\left[\bout = L(x) \mid \badv = 1\right] \geq 9/10 \, .
		\]
	\end{itemize}
	$L$ is in the subclass $\YQPstar$ if the family can be chosen such that
	$\badv$ is independent of $x$.
\end{definition}

Note that in the above definition, the subcircuit $B_n$ acts on the output of subcircuit $A_n$.
It does not necessarily receive a clean copy of the input.

The classes $\YQPpoly$ and $\YQPstarpoly$ 
are simply defined
by removing the requirement that the circuit family $\{Y_n\}_{n\in \N}$ be uniform \cite{aaronson2014full}.

Just as $\class{Oblivious\mhyphen{}NP}$ is unlikely to contain $\NP$ \cite{fortnow2009fixed},
it also seems unlikely that $\QMA$ is contained in $\YQP$.
We have the trivial bounds $\BQP\subseteq \YQPstar \subseteq \YQP\subseteq \QMA$
and $\YQP\subseteq \BQPqpoly$.
Studying $\YQP$ may be motivated by the use of oblivious complexity classes
in constructing circuit lower bounds \cite{fortnow2009fixed,chakaravarthy2006oblivious,GLV24oblivious},
by the fact that $\BQPqpoly = \YQPstarpoly = \YQPpoly$ shown %
in \cite{aaronson2014full}, %
or by the results shown in this %
article.  %

\paragraph{\texorpdfstring{$\mathsf{\mathbf{APP}}$}{APP}}
The class $\APP$ was introduced by
Lide %
Li \cite{li1993counting} in pursuit of
a large class of $\PP$-low languages.
We use the equivalent definition given by Fenner \cite[Corollary 3.7]{fenner2003pp}.

\begin{definition}\label{def:APP}
	$L\in \APP$ if and only if there exist
	functions $f,g\in \GapP$ and constants $0\leq \lambda < \upsilon \leq 1$ such that
	for all $n$ and $x\in\{0,1\}^n$,
	we have $g(1^n)>0$ and
	\begin{itemize}
		\item If $x\in L$ then
		$\upsilon g(1^n) \leq f(x) \leq g(1^n)$;
		\item If $x\notin L$ then $0\leq f(x) \leq \lambda g(1^n)$.
	\end{itemize}
\end{definition}

In the above definition, recall that $\GapP$ is the closure of $\sharpP$ under
subtraction.
In other words, while every function $f\in \sharpP$ corresponds to a nondeterministic
polynomial-time Turing Machine $N$ such that $f(x)$ equals the number of accepting paths
of $N(x)$,
a $\GapP$ function equals the number of accepting paths minus the number
of rejecting paths.

The class $\PP$ can be thought of as comparing a $\sharpP$ function
to a threshold exactly, with no promise gap.
The class in fact remains unchanged if it is defined as comparing a $\GapP$ function
to a threshold, and the threshold may be as simple as one-half of the possible paths or as complex as a
$\GapP$ function.
In these terms,
$\APP$ can be thought of as comparing a $\GapP$ function (here $f(x)$) to some threshold
(here $g(1^n)$),
where the complexity of the threshold is limited to a $\GapP$ function which may
depend on the input size but not the input,
and where there is some arbitrarily small but nonzero promise gap
(from $\lambda g(1^n)$ to $\upsilon g(1^n)$).
Comparing the two classes,
$\APP$ is a subclass of $\PP$ and is $\PP$-low, meaning $\PP^\APP = \PP$.

Like $\APP$, 
the best upper bound on the
well-known class $\AzeroPP=\class{SBQP}$ \cite{kuperberg2015hard}
is $\PP$,
so we make a brief comparison.
$\AzeroPP$ contains $\QMA$ \cite{vyalyi03}
and so also contains $\YQPstar$,
and $\AzeroPP$ is \emph{not} known to be $\PP$-low.
In contrast, $\APP$ is not known to contain even $\NP$ and is $\PP$-low.
Both $\APP$ and $\AzeroPP$ contain the class $\AWPP$ \cite{fenner2003pp,vyalyi03}.
However,
neither $\APP$ or $\AzeroPP$ is known to contain the other.

\paragraph{A useful fact}
We use the following fact shown for uniform circuit families by Watrous
\cite[Section~IV.5]{watrous2008quantum}, and shown earlier for QTMs by Fortnow and
Rogers \cite{fortnow1999complexity}.

\begin{lemma}\label{lem:quantumGapP}
	For any polynomial-time uniformly generated family of quantum circuits
	$\{Q_n\}_{n\in \N}$ each of
	size bounded by a polynomial $t(n)$,
	there is a $\GapP$ function $f$ such that for all $n$-bit $x$,
	\[
	\Pr\left[Q_n(x)\text{ accepts}\right] = \frac{f(x)}{5^{t(n)}} \, .
	\]
\end{lemma}

\paragraph{Error reduction}

In our proof that $\YQPstar\subseteq \APP$, we perform error reduction on quantum circuits.
Error reduction for complexity classes involving quantum inputs, such as $\QMA$ or $\BQPqpoly$,
is often performed using many copies of the input in parallel,
but we wish to use a particular state $\rho$ as input, not $\rho^{\otimes k}$.
Therefore,
we require the ``in-place'' error reduction technique of
Marriott and Watrous \cite{marriott2005quantum}.

\begin{theorem}[Theorem 3.3 of \cite{marriott2005quantum}]\label{thm:MW}
	Let $a,b:\N\rightarrow [0,1]$ and $q\in \poly(n)$ satisfy $a(n)-b(n)\geq 1/q(n)$.
	Consider any quantum circuit $C$
	of size $s$
	acting on an $m$-qubit input and an all-zero ancilla register
	such that $C$ accepts with probability at least $a$ or at most $b$ %
	for $s,m\leq\poly(n)$.
	Then there exists a polynomial-time procedure
	that, for any $r\in\poly(n)$,
	produces a circuit $C'$
	of size $\poly(s)$ that also
	acts on an $m$-qubit input
	and accepts with probability at least %
	$1-2^{-r}$ or at most $2^{-r}$.
\end{theorem}

Our proof requires analyzing not just the output probabilities of the amplified circuit,
but also the state of the system at the end of the circuit.
Fortunately, studying the proof of \cite[Theorem 3.3]{marriott2005quantum} provides some straightforward observations.
Here we briefly sketch part of the construction and then state the properties necessary for our proof.

Consider some quantum circuit $C$ that takes an all-zero ancilla register $\ket{0\dots 0}$ and some quantum state as input
and that accepts or rejects with some probability.
The error reduction algorithm of \cite{marriott2005quantum} involves
applying the circuit $C$,
measuring the output qubit and
recording whether it is $\ket{0}$ or $\ket{1}$ in a variable $y_{2i-1}$,
applying $C^{\dagger}$, 
measuring the circuit's ancilla register and
recording whether it is
in the all-zero
state or not in a variable $y_{2i}$, and repeating these steps for some number of
iterations $M$.
Call the full, amplified circuit $C'$.
At the end of $C'$, the recorded bits $\{y_1,\dots,y_{2M}\}$ can be used to estimate the probability $C$ accepts
with high precision. 
Specifically, the more pairs such that $y_i = y_{i+1}$, the more likely that $C'$ accepts.

The recorded bits also tell us about the state of the system after applying $C'$.
If after applying $C^\dagger$, a bit $y_{2i}=1$,
then the state of the ancilla register was projected into the all-zero state.
Now, suppose the circuit $C'$ is applied to an $m$-qubit proof state,
so there are $2^m$ eigenstates $\{\ket{\lambda_i}\}_{j\in [2^m]}$ of $C'$.
Studying the proof of \cite{marriott2005quantum},
if the initial state given to $C'$ was an eigenstate $\ket{\lambda_j}$,
and after a round of applying $C$ and $C^\dagger$ the recorded bits were $y_{2i-1}=y_{2i}=1$,
then not only is the ancilla register known to be in the all-zero state,
but the final state of the proof register is the same as its initial state,
$\ket{\lambda_j}$.

A brief analysis allows us to characterize the probability of this outcome.
Note $C$ and $C'$ have the same $m$-qubit eigenstates,
and suppose an eigenstate $\ket{\lambda_j}$ is accepted by the
original circuit $C$ with probability $p$,
Intuitively, consecutive bits are transitions which depend on whether we expect
the circuit beginning with a particular proof and properly initialized all-zero ancilla register to produce an output qubit close to $\ket{1}$ or to $\ket{0}$, and vice-versa.
In other words, when $C'$ is run on $\ket{\lambda_j}$,
we have $\Pr\left[y_{i}=1 \mid y_{i-1}=1\right] = p$ for all $i$.
This is a two-state Markov chain
with probability $p$ of changing states.
Raising the appropriate transition matrix to the $i$-th power
and applying it to the initial state $y_0=1$ ($C'$ begins with ancilla in the all-zero state) 
yields $\Pr\left[y_i=1\right] = \left(\left(2p-1\right)^i + 1\right)/2$ for all $i>0$.
If we make the additional assumption that $p>1/2$, then this probability is greater than $1/2$.
Then we can also conclude that $\Pr\left[y_{i+1}=y_i=1\right] > p/2$ for all $i$.
Moreover, as noted above,
any pair $y_{i}=y_{i+1}$ only increases the probability the amplified circuit accepts, allowing us to %
calculate a lower bound on the joint probability. %

\begin{lemma}[Extension of Theorem 3.3 of \cite{marriott2005quantum}]\label{lem:MW2}
	In addition to the statement of \expref{Theorem}{thm:MW},
	the amplified circuit
	$C'$ records two final variables $y,z\in \{0,1\}$.
	Suppose $\ket{\lambda}$ is an eigenstate of $C$
	and that
	$C'$ is run on $\ket{\lambda}$.
	If
	$y=z=1$,
	then the system is left in the state $\ket{\lambda}\ket{0\dots 0}$.
	If $C$ accepts $\ket{\lambda}$ with probability at least $a$ and $a > 1/2$,
	then the probability that $C'$ accepts $\ket{\lambda}$ and the variables $y=z=1$ is at least $\left(1-2^{-r}\right) a /2$.
\end{lemma}

\section{Results}

We first prove our main technical result, that $\YQPstar\subseteq \APP$.
Our approach is as follows.
$\APP$ evaluates the ratio of two $\GapP$ functions, where one of the functions
is only allowed to depend on the input length. By \expref{Lemma}{lem:quantumGapP}, functions
in $\GapP$ can encode the output probabilities of quantum circuits.
So, for a $\YQPstar$ computation with circuit $Y_n$ and subcircuit $A_n$,
we run them both on a random proof using the maximally mixed state
and ask $\APP$ to determine the ratio of their acceptance probabilities.
This is possible because the acceptance probability of $A_n$ over a random proof
depends only on the input length.
We perform error reduction on the subcircuit $A_n$ so that
$A_n$ mistakenly accepting ``bad'' proofs has a negligible effect on the acceptance probability of $Y_n$.
Thus, the ratio of probabilities approximates how often
$Y_n$ accepts given a good proof.
Because we wish to use the maximally mixed state as input
and parallel copies of a uniform mixture
is not the same as a uniform mixture of parallel copies,
we require the ``in-place'' error reduction technique reviewed above.

\begin{lemma}\label{lem:YQPinAPP}
	$\YQPstar\subseteq \APP$.
\end{lemma}
\begin{proof}
	Consider any language $L\in \YQPstar$.
	Let $\{Y_n,A_n,B_n\}_{n\in\N}$ be the associated family of circuits
	and subcircuits,
	in which $Y_n$ takes string $x$ and a supposed proof or advice state as input,
	in which subcircuit
	$A_n$ validates the proof and produces output bit $\badv$,
	and in which,
	given $A_n$ accepted, $B_n$ uses the proof to verify whether the particular input
	$x$ is in $L$, producing the output bit $\bout$.
	Note that because we consider $\YQPstar$, the circuit $A_n$ only takes
	the proof state, not $x$, as input.
	Let $m$ be a polynomial in $n$ denoting the size of the
	proof register.
	
	Apply the in-place error reduction of \cite{marriott2005quantum} 
stated in \expref{Theorem}{thm:MW}
	on the circuits $A_n$
	with a polynomial $q$ in $n$ of our choosing
	to produce a new circuit family $\{A'_n\}_{n\in\N}$
	such that for any proof $\rho$,
	\begin{itemize}
		\item $\Pr\left[A_n(\rho)\right] \geq \frac{9}{10} \Rightarrow
		\Pr\left[A'_n(\rho)\right]\geq 1-2^{-q}$;
		\item $\Pr\left[A_n(\rho)\right] \leq \frac{1}{10} \Rightarrow
		\Pr\left[A'_n(\rho)\right]\leq 2^{-q}$.
	\end{itemize}
	For later use, we choose $q \geq \max\{2m,10\}$.
	Note that $A'_n$ also produces the two variables $y$ and $z$ described by \expref{Lemma}{lem:MW2}.
	
	We define $\{A''_n\}_{n\in\N}$ to be the amplified
	circuits $\{A'_n\}_{n\in\N}$ with the additional rule that the circuit
	accepts iff both $\badv=1$ and the two recorded bits $y=z=1$.
	Further, define $\{A'''_n\}_{n\in\N}$ so that $A'''_n = A''_n(\frac{\bbI}{2^m})$,
	with the maximally mixed state hard-wired into the proof register.
	Similarly, we define $\{Y'_n\}_{n\in\N}$ to apply the amplified subcircuit
	$A'_n$ and $B_n$,
	we define $\{Y''_n\}_{n\in\N}$ to apply $A''_n$ and $B_n$ and thus accept iff
	$\badv,\bout,y,z$ all equal~1,
	and we define $\{Y'''_n\}_{n\in\N}$ so that $Y'''_n(x) = Y''_n(x,\frac{\bbI}{2^m})$
	with the maximally mixed state hard-wired into the proof register,
	meaning that it uses $A'''_n$ as a subcircuit.
	
	\paragraph{Analysis}
	Applying \expref{Lemma}{lem:quantumGapP}, there exist $\GapP$ functions $f,g$ and polynomials
	$r,t$ such that for all $n$-bit $x$,
	\begin{gather*}
		\Pr\left[A'''_n\text{ accepts}\right] = \frac{f(1^n)}{5^{r(n)}}
		\quad\text{and}\quad
		\Pr\left[Y'''_n(x)\text{ accepts}\right] = \frac{g(x)}{5^{t(n)}} \, .
	\end{gather*}
	The function $f$ depends only on the input length $n$, not $x$,
	because the circuit $A'''_n$ is independent of $x$.
	Next, we define $F(1^n) = f(1^n) 5^{t(n)-r(n)}$, which is a $\GapP$ function
	since $5^{t(n)-r(n)}\in \class{FP}\subseteq \GapP$ and $\GapP$ is closed under
	multiplication.
	Given the definition of $\YQPstar$ guarantees there
	exists a ``good'' proof for circuit $A_n$,
	we have  $f(1^n),F(1^n)>0$.
	Combining these definitions,
	\begin{gather*}
		\frac{g(x)}{F(1^n)} =
		\frac{\Pr\left[Y'''_n(x)\text{ accepts}\right]}
		{\Pr\left[A'''_n\text{ accepts}\right] } \, .
	\end{gather*}
	
	We will show bounds on the ratio $g(x)/F(1^n)$ based on whether $x$ is in $L$ or
	not in $L$ in order to prove $L$ is in $\APP$.
	First, note that the ratio is
 at most %
 1 since $Y'''_n$ only accepts if the subcircuit $A'''_n$ accepts,
 and it is %
 at least %
0 since probabilities are non-negative.
	Next, let $\left\{\ket{\lambda_i}\right\}_{i\in[2^m]}$ be the set of
	eigenvectors of the circuit $A_n$.
	By writing the maximally mixed state, which is hard-wired into the proof register of
	$Y'''_n$,
	in terms of this eigenbasis, we find
	\[
		\frac{\Pr\left[Y'''_n(x)\text{ accepts}\right]}
		{\Pr\left[A'''_n\text{ accepts}\right]}
		= \frac{\Pr\left[Y''_n(x,\frac{\bbI}{2^m})\text{ accepts}\right]}
		{\Pr\left[A''_n(\frac{\bbI}{2^m})\text{ accepts}\right]}
		= \frac{\sum_{i=1}^{2^{m}}
			\Pr\left[Y''_n(x,\ket{\lambda_i})\text{ accepts}\right]}
		{\sum_{i=1}^{2^{m}} \Pr\left[A''_n(\ket{\lambda_i})
			\text{ accepts}\right]} \, .
	\]
	Next, use the fact that $Y''_n$ accepting requires that
	$A''_n$ accepts to find the above equals
	\[
		\frac{\sum_{i=1}^{2^{m}}
		\Pr\left[Y''_n(x,\ket{\lambda_i})\text{ accepts}
		\mid A''_n(\ket{\lambda_i})\text{ accepts} \right]
		\Pr\left[A''_n(\ket{\lambda_i})\text{ accepts}\right]}
		{\sum_{i=1}^{2^{m}} \Pr\left[A''_n(\ket{\lambda_i})
		\text{ accepts}\right]}  \, .
	\]
	The observation from \expref{Lemma}{lem:MW2}
	implies	that if the amplified circuit $A''_n$ %
	accepts, then %
	the state sent on to the subcircuit $B_n$ within $Y''_n$ is 
the initial eigenstate $\ket{\lambda_i}$.
	Then, the above equals
	\[
		\frac{\sum_{i=1}^{2^{m}}
			\Pr\left[B_n(x,\ket{\lambda_i})\text{ accepts}\right]
			\cdot \Pr\left[A''_n(\ket{\lambda_i})\text{ accepts}\right]}
		{\sum_{i=1}^{2^{m}} \Pr\left[A''_n(\ket{\lambda_i})
			\text{ accepts}\right]} \, .
	\]

	Define
	\[
	\calB = \left\{ i\in[2^m] \mid \Pr\left[A_n(\ket{\lambda_i})\right]
	\leq 0.1 \right\} \, ,
	\]
	which are intuitively the ``bad'' proofs, such that states in $\calB$ will be rejected
	by $A'_n$ with high probability while the ``not bad'' states in $\overline{\calB}$
	cause $B_n$ to output the correct answer with high probability.
	We can now rewrite the numerator in the above ratio as
	\begin{multline*}
		\sum_{i\in \calB}
		\Pr\left[B_n(x,\ket{\lambda_i})\text{ accepts}\right]
		\cdot \Pr\left[A''_n(\ket{\lambda_i})\text{ accepts}\right] 
		\\ + 
		\sum_{i\in\overline{\calB}} 
		\Pr\left[B_n(x,\ket{\lambda_i})\text{ accepts}\right]
		\cdot \Pr\left[A''_n(\ket{\lambda_i})\text{ accepts}\right] 
	\end{multline*}
	and rewrite the denominator as
	\[
		\sum_{i\in \calB} \Pr\left[A''_n(\ket{\lambda_i})\text{ accepts}\right]
		+ \sum_{i\in\overline{\calB}} \Pr\left[A''_n(\ket{\lambda_i})\text{ accepts}\right]
		\, .
	\]
	We will use this expression for $g(x)/F(1^n)$ as the starting point in our analysis of the
	YES and NO cases.
	Additionally, let $\ket{\lambda^*}$ denote a proof in $\overline{\calB}$,
	which is guaranteed to be nonempty by the definition of $\YQPstar$.
	By the observation from \expref{Lemma}{lem:MW2},
	$\Pr\left[A''_n(\ket{\lambda^*})\text{ accepts}\right] \geq \left(1-2^{-q}\right) (0.9)(0.5)$.

	Suppose we have a YES instance with $x\in L$.
	Then, we may calculate that $g(x)/F(1^n)$ is at least
	\begin{align*}
		\frac{
			\sum_{i\in \calB} 0
			+ \sum_{i\in\overline{\calB}} \frac{9}{10}
			\Pr\left[A''_n(\ket{\lambda_i})\text{ accepts}\right]
		}
		{
			\sum_{i\in \calB} 2^{-q}
			+ \sum_{i\in\overline{\calB}}
			\Pr\left[A''_n(\ket{\lambda_i})\text{ accepts}\right]
		}
		&= \frac{
			\sum_{i\in\overline{\calB}} \frac{9}{10}
			\Pr\left[A''_n(\ket{\lambda_i})\text{ accepts}\right]
		}
		{
			\abs{\calB} 2^{-q}
			+ \sum_{i\in\overline{\calB}} \Pr\left[A''_n(\ket{\lambda_i})\text{ accepts}\right]
		} \\
		&\geq \frac{\frac{9}{10} \Pr\left[A''_n(\ket{\lambda^*})\text{ accepts}\right]}
		{ \abs{\calB} 2^{-q} + \Pr\left[A''_n(\ket{\lambda^*})\text{ accepts}\right]} \, ,
	\end{align*}
	where the second line follows by the fact that $x/(c+x)$ decreases as $x$ decreases.
	Applying this fact again along with our choice $q\geq\max\{2m,10\}$, we find
	the above is at least
	\begin{align*}
		\frac{\frac{9}{10}(1-2^{-q}) (0.9)(0.5) }
		{\abs{\calB}2^{-q} + (1-2^{-q}) (0.9)(0.5)}
		&\geq \frac{ 0.405 (1-2^{-q}) }{2^{m-q} + 0.45(1-2^{-q}) } \\
		&\geq \frac{0.405(1-2^{-q}) }{2^{q/2-q} + 0.45(1-2^{-q}) } > 0.84 \, .
	\end{align*}
	
	On the other hand, consider a NO-instance.
	By similar steps as in the YES-case,
	we have that $g(x)/F(1^n)$ is at most
	\begin{align*}
		\frac{
			\abs{\calB} 2^{-q} + \sum_{i\in\overline{\calB}}
			\frac{1}{10} \Pr\left[A''_n(\ket{\lambda_i})\text{ accepts}\right]
		}
		{
			\sum_{i\in \calB} 0 + \sum_{i\in \overline{\calB}} \Pr\left[A''_n(\ket{\lambda_i})\text{ accepts}\right]
		}
		&\leq
		\frac{2^{m-q}}
		{\sum_{i\in \overline{\calB}} \Pr\left[A''_n(\ket{\lambda_i})
			\text{ accepts}\right]
		}
		+
		\frac{1}{10} \\
		&\leq
		\frac{2^{-q/2}}{\Pr\left[A''_n(\ket{\lambda^*})\text{ accepts}\right]} + \frac{1}{10}	\\
		&
		\leq \frac{2^{-q/2} }{\left(1-2^{-q}\right) (0.9) (0.5)} + \frac{1}{10}
		< 0.2	\, .
	\end{align*}
	
	We have shown a constant separation of $g(x)/F(1^n)$ in YES- and NO-instances.
	This satisfies the criteria of $\APP$ in \expref{Definition}{def:APP},
	so we conclude $\YQPstar\subseteq \APP$.
\end{proof}

Next, the fact $\APP$ is known to be $\PP$-low \cite[Theorem 6.4.14]{li1993counting}
gives us the following corollary.

\begin{corollary}\label{corr:YQPisLow}
	$\YQPstar$ is $\PP$-low, \ie{}, $\PP^{\YQPstar}=\PP$.
\end{corollary}

For intuition, an alternative proof of \expref{Corollary}{corr:YQPisLow} without \expref{Lemma}{lem:YQPinAPP} 
might have relied
on the equality $\PP=\postBQP$ \cite{aaronson2005quantum},
where $\postBQP$ has the ability to post-select, \ie{}, it is guaranteed to output
the correct answer with high probability \emph{conditioned on} some other event
which may occur with very small probability.
So, instead of $\PP^\YQPstar$, we might have considered $\postBQP^\YQPstar$.
Whenever the $\postBQP$ machine would make a query, it instead could
run the $\YQPstar$ proof-validation circuit on
the maximally mixed state, post-select on it accepting,
then simulate the rest of the $\YQPstar$ computation.

We are now able to give a corrected proof of the result originally claimed for
$\BQPqpoly$ but only proved for $\BQPpoly$ by Aaronson \cite{aaronson_subtle}.
We mostly repeat Aaronson's proof, but substitute $\YQPstar$ where he relied on $\QMA$.

\begin{theorem}\label{thm:main}
	If $\PP \subseteq \BQPqpoly$, then the Counting Hierarchy collapses to
	$\CH = \QMA = \YQPstar$.
\end{theorem}

\begin{proof}
	Suppose $\PP \subseteq \BQPqpoly$.
	Clearly then $\PTIME^\PP$ is also contained in $\BQPqpoly$.
	So, the $\sharpP$-complete problem \perm{} is contained in $\BQPqpoly$, since $\sharpP\subseteq \PTIME^\sharpP=\PTIME^\PP$.
	From \cite{aaronson2014full}, we know that
	$\BQPqpoly=\YQPstarpoly$.
	A $\YQPstarpoly$ protocol involves a circuit, a trusted advice string, and an untrusted quantum proof.
	Let $C$ and $a$ be the circuit and advice string for solving \perm{} in $\YQPstarpoly$.
	
	Next in $\YQPstar$, without trusted advice, 
	Arthur can request that an untrusted Merlin
	send many copies of the resources from the above $\YQPstarpoly$ protocol,
	including the quantum proof
	and a supposed copy of the circuit $C$ and advice $a$.
	To check these \emph{untrusted} copies of the proof, circuit, and advice,
	Arthur generates a random set of inputs and
	simulates the interactive protocol for \perm{} due to \cite{lund1992algebraic}
	using the copies in place of the prover.
	If the protocol accepts (meaning the ``prover'' worked),
	then with high probability, Merlin must have sent resources that work on a large fraction of inputs.
	This ends the proof-verification phase of $\YQPstar$.
	
	Given a circuit and advice that work on most inputs,
	Arthur can use the random self-reducibility of
	\prob{Permanent} to generate a circuit $C'$ that is correct on \emph{all} inputs with high probability
	(see \eg{}, \cite[Sec. 8.6.2]{arora2009computational}).
	Thus, after the verification phase, the $\YQPstar$ protocol can simulate $\sharpP$.
	By the same argument as above, this also means it can simulate $\PTIME^\sharpP \supseteq \PP$, 
	and we have $\PP = \YQPstar$.
	
	In this way, any level of the Counting Hierarchy
	$\CHi[i] = \left(\CHi[i-1]\right)^\PP$ with $i>1$ is reducible to
	$\left(\CHi[i-1]\right)^{\YQPstar}$
	which by \expref{Corollary}{corr:YQPisLow} equals $\CHi[i-1]$.
	This works recursively for all levels, collapsing $\CHi[i]$ to $\CHi[1]=\PP$,
	so that all of $\CH=\PP=\YQPstar$.
\end{proof}

An alert reader may notice the proof never directly asks the base $\PP$ machine to
guess advice and hard-code it into a query, as referenced in the introduction.
That step was explicit in the original proof %
in \cite{aaronson_subtle}, %
to reduce $\PP^\BQPpoly$ to $\PP^\BQP=\PP$.
Here, unlike $\BQP$, the class $\YQPstar$ is itself able to guess a proof,
so we reduce $\PP^\YQPstarpoly$ to $\PP^\YQPstar$ without ``asking'' anything of the
base computation until collapsing $\PP^\YQPstar$ to $\PP$.

Given the above result, we can also fully recover the following result originally
claimed by Aaronson \cite{aaronson_subtle}, giving an improved unconditional upper bound
on %
fixed-polynomial-size  %
quantum circuits with quantum advice.
For completeness, we repeat the proof.

\begin{theorem}\label{thm:PPcircs}
	$\PP$ does not have quantum circuits of size $n^k$ for any fixed $k$.
	Furthermore, this holds even if the circuits can use quantum advice.
\end{theorem}
\begin{proof}
	Suppose $\PP$ does have circuits of size $n^k$.
	This implies $\PP\subseteq \BQPqpoly$, which by \expref{Theorem}{thm:main} implies
	$\CH = \YQPstar$, which includes
	$\PTIME^\PP = \PP = \YQPstar$.
	Together, there are circuits of size $n^k$ for $\PTIME^\PP$, which contradicts
	the result \cite[Theorem 4]{aaronson_subtle} (unaffected by the bug)
	that $\PTIME^\PP$ does not have such circuits even with quantum advice.
\end{proof}

In fact, \cite{aaronson_subtle} noted
that the proof showing
$\PTIME^\PP$ does not have circuits of size $n^k$ for fixed $k$
even with quantum advice can be strengthened.
Substituting this stronger result into the above proof,
we have that \expref{Theorem}{thm:PPcircs} can be strengthened to show for all functions
$f(n)\leq 2^n$,
the class $\class{PTIME(f(f(n)))}$,
which is like $\PP$ but for machines of running time $f(f(n))$,
requires quantum circuits using quantum advice of size at least $f(n)/n^2$.
In particular, this implies $\class{PEXP}$, the exponential-time version of $\PP$,
requires quantum circuits with quantum advice of ``half-exponential'' size (meaning
a function that becomes exponential when composed with itself \cite{miltersen1999super}).

\bibliographystyle{tocplain}  

\bibliography{bibstrings,v021a007,bibtail}

\begin{tocauthors}
\begin{tocinfo}[yirka]
 Justin Yirka\\
 Quantum Computing Researcher\\ %
 Blanqet\\
 Austin, TX, USA\\
 justinyirka\tocat{}gmail.com\\ %
 \url{https://justinyirka.com/}      %
\end{tocinfo}
\end{tocauthors}

\begin{tocaboutauthors}
\begin{tocabout}[yirka]  %
  \textsc{Justin Yirka} graduated from
  \href{https://www.utexas.edu/}{The University of Texas at Austin}
  in May 2025, advised by
  \href{https://scottaaronson.com/}{Scott Aaronson}.
  The subject of his dissertation was quantum computational 
  complexity, including the results in this article. 
  He is now a quantum computing researcher for \href{http://blanqet.net/}{Blanqet}, 
  a startup founded by professors from
  \href{https://www.uchicago.edu/}{the University of Chicago}
  and where his \phd{} adviser, Scott Aaronson, is a partner.
  Blanqet's office is in Chicago and Justin works remotely from Austin.
  Justin earned his \bsc{} from \href{https://www.vcu.edu/}{Virginia Commonwealth University}, 
  where he was introduced to research by working with \href{https://groups.uni-paderborn.de/fg-qi/}{Sevag Gharibian}.
  Justin has also worked at \href{https://www.sandia.gov/}{Sandia National Laboratories}, at \href{https://www.lanl.gov/}{Los Alamos National Laboratory}, 
as an Uber driver, and as a lifeguard.
\end{tocabout}
\end{tocaboutauthors}

\end{document}